\documentclass[a4paper,aps,PRA,showpacs,longbibliography,twocolumn,superscriptaddress]{revtex4-1}   
   
\usepackage[utf8]{inputenc}  
\usepackage[T1]{fontenc}     
\usepackage[british]{babel}  
\usepackage{lmodern}  
\usepackage{hhline}

\usepackage[scaled=1.03]{inconsolata} 
\usepackage[usenames,dvipsnames]{color} 
\usepackage[colorlinks,citecolor=blue,linkcolor=magenta,urlcolor=blue]{hyperref}  
\usepackage{graphicx} 
\usepackage{tikz}
\usetikzlibrary{shapes.geometric,shadows.blur, decorations.pathreplacing, arrows, decorations, positioning, hobby, decorations.text}

\usepackage[babel]{microtype}  
\usepackage{amsmath,amssymb,amsthm,bm,mathtools,amsfonts,mathrsfs,bbm,dsfont} 
\usepackage{xspace}  
\usepackage{multirow}
\usepackage{verbatim}
\usepackage{enumerate}
\usepackage{tcolorbox}
\usepackage[normalem]{ulem}

\usepackage{physics}
\newcommand{\id}{\ensuremath{\mathds{1}}}
\usepackage{bbold}
\usepackage[mathscr]{eucal}

\newcommand{\Dyn}{\mathcal{D}}

\newcommand{\wit}{s_{\text{QM}}}

\newcommand{\Inst}{\mathcal{I}}

\newcommand{\inst}{\mathcal{I}}

\usepackage{braket}

\def\i{\ensuremath{\mathrm{i}}}

\newcommand{\cpt}{\mathcal{E}}
\newcommand{\dynamics}[1]{(#1)}

\newtheorem{theorem}{Theorem}
\newtheorem{corollary}{Corollary}

\usepackage{bbold}

\newcommand{\cptext}{\cpt^*}
\newcommand{\cptu}{\cpt_u}

\newcommand{\cptls}{\cpt_{\text{LS}}}
\newcommand{\cptao}{\cpt_{\text{AO}}}
\newcommand{\cpthmr}{\cpt_{\text{HMR}}}
\newcommand{\cptd}{\cpt_{\text{D}}}

\newcommand{\sys}{\mathrm{sys}}
\newcommand{\env}{\mathrm{env}}

\usepackage{tikz}
\usetikzlibrary{patterns}
\usetikzlibrary{patterns.meta}

\usetikzlibrary{arrows.meta, shadings}

\usepackage{pgfplots}
\usepackage{tikz-3dplot}
\tdplotsetmaincoords{60}{115}
\tdplotsetrotatedcoords{25}{-25}{0}
\pgfplotsset{compat=newest}
\usetikzlibrary{arrows.meta, arrows, positioning}

\pgfplotsset{compat = newest}

\definecolor{mathematicablue}{rgb}{0.87,0.94,1}
\definecolor{mathematicadarkblue}{rgb}{0.368417, 0.506779, 0.709798}
\definecolor{mathematicaorange}{rgb}{1,0.9,0.8}
\definecolor{mathematicadarkorange}{rgb}{1,0.5,0}

\usepackage[dvipsnames]{xcolor}

\makeatletter
\def\maketitle{
\@author@finish
\title@column\titleblock@produce
\suppressfloats[t]}
\makeatother

\begin{document}

 \author{Charlotte Bäcker}
 \affiliation{Institute of Theoretical Physics, TUD Dresden University of Technology, 01062, Dresden, Germany}
 \author{Konstantin Beyer}
 \affiliation{Department of Physics, Stevens Institute of Technology, Hoboken, New Jersey 07030, USA}
 \author{Walter T. Strunz}
 \affiliation{Institute of Theoretical Physics, TUD Dresden University of Technology, 01062, Dresden, Germany}

\title{Quantum memory precludes mixed-unitary dynamics}
\date{\today}

\begin{abstract}
    Unital quantum channels, defined by their property of leaving the maximally mixed state invariant, form an important class of quantum operations. A distinguished subset of these channels can be represented as a probabilistic mixture of unitary evolutions.
    Characterizing channels that do not admit such a decomposition is in general a hard problem with significant implications for noise mitigation in quantum technologies and for fundamental problems in quantum information theory.
    Here we establish a link between mixed-unitarity of unital channels and the (quantum) nature of the memory effects in non-Markovian dynamics. 
    Translating the problem into the language of process tensors, this connection yields a hierarchy of semidefinite programs that provides numerically efficient witnesses for non-mixed-unitary behavior, outperforming existing criteria. We demonstrate the power of this approach through illustrative examples of unital channels in dimensions three and four, and show how the method can detect non-mixed-unitary behavior based on incomplete channel tomography.
\end{abstract}

\maketitle

\section{Introduction}

Unital quantum channels, i.e., those that preserve the identity, \(\cptu[\id]=\id\), play a central role in many areas of quantum dynamics, including resource theories~\cite{ChiGou2019ResourceTheories, GOUR20151, Streltsov_2018}, majorization and fluctuation theorems in quantum thermodynamics~\cite{Watrous_2018, ContinuousMajorization2025, Rastegin_2013, AlbashFluctuation2013}, and noise modeling in the context of quantum computing. The dominant decoherence processes in quantum devices are often unital, as they occur on much shorter time scales than non-unital damping channels~\cite{laddQuantumComputers2010,suterColloquiumProtectingQuantum2016,paladino1NoiseImplications2014,krantz2019}.
A particular class of unital channels are those of mixed-unitary (MU) type
\begin{align}
    \cpt_\text{MU}[\rho] = \sum_i p_i U_i \rho U_i^\dagger, && p_i \geq 0, && \sum_i p_i =1,
\end{align}
with the $U_i$ being unitary.
Such channels describe unitary quantum dynamics conditioned on classical noise, and can be seen as the result of averaging over a Hamiltonian with a classical stochastic parameter, for example a random external field~\cite{alickiQuantumDynamicalSemigroups2007,HelStrPRA2010, BudiniRandomLindblad}.

Finite-dimensional unital channels form a convex and compact set~\cite{MenWol2009}.
Thus, they can always be decomposed into the set's extreme points \(\cptext_i\) as \(\cptu = \sum_i p_i \cptext_i\) with \(p_i \geq 0,\, \sum_i p_i = 1\).
Any unitary is an extreme point and for Hilbert spaces of dimension $d=2$ all of these extreme elements are unitary~\cite{LanStr1993}.
This implies that any unital channel on a qubit system 
is MU.

It has long been known that this is no longer true for dimensions \(d\geq 3\), where unital channels exist that do not admit an MU representation~\cite{kummererEssentiallyCommutativeDilations1987, LanStr1993, AudSch2008, ohnoMaximalRankExtremal2010, HelStr2009, haagerupExtremePointsFactorizability2021}. Various special classes of such channels have been proposed in the literature~\cite{MenWol2009, HelStr2009, watrousMixingDoublyStochastic2009, ohnoMaximalRankExtremal2010, haagerupExtremePointsFactorizability2021, rodriguez-ramosConvexCharacterisationSet2023, kribsOperatorAlgebraGeneralization2024a,kummererEssentiallyCommutativeDilations1987}, see also Fig.~\ref{fig:geometry-qutrit-channels}. 
\begin{figure}
    \centering
    \includegraphics[]{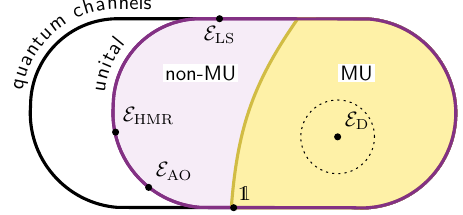}
    \caption{Schematic representation of the space of quantum channels highlighting the five unital qutrit channels investigated in this article. Two of them are MU, the other three are non-MU. The set of unital non-MU channels is not convex.}
    \label{fig:geometry-qutrit-channels}
\end{figure}
However, generally applicable tools for the detection of non-MU unital channels are lacking and a broader characterization of such channels remains elusive. 
In fact, deciding whether a given unital channel is MU is NP-hard in the sense of a Cook reduction~\cite{leeDetectingMixedunitaryQuantum2020, cookComplexityTheoremprovingProcedures1971}. 
Nevertheless, various necessary and sufficient criteria~\cite{AudSch2008}, as well as methods based on distance measures~\cite{AudSch2008,HelStr2009} and witnesses~\cite{LanStr1993,MenWol2009} have been proposed that allow for the detection of non-MU channels. In our current work we add to this by rephrasing the problem in form of a semidefinite program.
Beyond mere academic interest, this question is of crucial practical relevance: MU channels always admit an environment-assisted error correction scheme that reverses the channel~\cite{GreWer2003, TS2011EnvironmentAssistedCorrection}. This means that decoherence processes described by MU dynamics can, in principle, be mitigated perfectly. For non-MU unital channels, no such correction scheme can exist~\cite{GreWer2003}. This is sometimes called \emph{truly quantum decoherence}, in contrast to a classical statistical mixture of unitary operations \cite{HelStr2009}.

Recent developments have shown that qutrits, ququarts, and even higher-dimensional quantum systems may be beneficial for quantum computing \cite{lanyonSimplifyingQuantumLogic2009, fedorovImplementationToffoliGate2012, gedikComputationalSpeedupSingle2015, ringbauerUniversalQuditQuantum2022, TanayTwoQutritAlgorithms2023, PavlidisQFTQudit2021, SeifertQuquarts2023}, making the distinction between MU and non-MU noise processes directly relevant for error models and mitigation strategies. Moreover, in qubit-based devices, noise processes that act jointly on multiple qubits become increasingly relevant as single-qubit noise levels improve.

In this article, we address the problem of detecting non-MU unital channels from an at first sight unrelated perspective. We map this problem to a recently developed framework for quantum memory in non-Markovian dynamics~\cite{BaeBeyStr2024}. The key insight is that since any MU dynamics can be reversed by an environment-assisted scheme, such a reversal never necessitates quantum memory, as we will show below. Accordingly, witnesses of quantum memory in non-Markovian dynamics can be directly used to demonstrate non-MU behavior.

Building on existing criteria for quantum memory~\cite{YuOhsNguNim2025:p}, we provide numerical efficient methods, based on semidefinite programs (SDPs), for detecting non-MU channels. 
We benchmark these by showing that our witnesses outperform or match known criteria for various classes of qutrit non-MU channels.
We then investigate a time-continuous joint dephasing process of a two-qubit system and discuss how its non-MU nature can be demonstrated experimentally without full tomography.

\section{Mapping to a quantum memory-related problem}
The field of non-Markovian quantum dynamics has long been an area of intense research. Recently, especially the origin of the memory effects has attracted much interest~\cite{MilEglTarThePleSmiHue2020, GiaCos2021, BanMarHorHor2023, TarQuiMurMil2023:p, BaeBeyStr2024, BusGanGosBadPanMohDasBer2025,  goswamiHamiltonianCharacterizationMultitime2025, BaeBeyStr2025, YuOhsNguNim2025:p, BaeLinStr2025}.
Following the definition in Refs.~\cite{BaeBeyStr2024,YuOhsNguNim2025:p} a two-step dynamics $\Dyn=\dynamics{\cpt_1, \cpt_2}$ of two completely positive trace-preserving (CPT) maps can be realized with a \emph{classical memory} if these maps can be decomposed as 
\begin{align}
    \label{eq:def_classicalmemory_instruments}
        \cpt_2 = \sum_\alpha \mathcal{F}_\alpha \circ \Inst_\alpha, \qquad \text{with } \sum_\alpha \Inst_\alpha = \cpt_1,
\end{align}
where $\mathcal{F}_\alpha$ is a family of CPT maps and the $\Inst_\alpha$ form a quantum instrument~\cite{heinosaariMathematicalLanguageQuantum2011}.
The reasoning behind this definition is that the dynamics can be decomposed into a measurement described by \(\Inst_\alpha\) for the first map \(\cpt_1\), followed by a CPT channel \(\mathcal{F_\alpha}\) that is conditioned on the \emph{classical} measurement outcome \(\alpha\) (see Fig.~\ref{fig:evolution}).
\begin{figure}
    \centering
	\begin{tikzpicture}
	\node (1) at (-2, -0.25) {$\rho_{t_0}$};
	\node (2) at (0, -0.25) {$\rho_{t_1}$};
	\node (3) at (1.9, -0.2) {$\rho_{t_2}$};
	\node (4) at (-1, 1) {$\cpt_1$};
	\node (4) at (1, 1) {};
	\node (4) at (0, -1.4) {$\cpt_2$};
	\node [] (0) at (-2, -0) {};
	\node [] (1) at (0, -0) {};
	\node [] (2) at (2, -0) {};
	\node [] (3) at (-2, -0.5) {};
	\node [] (4) at (2, -0.5) {};
	\draw [->, line width=0.4mm, bend left=75, looseness=1.10] (0.center) to (1.center);
	\draw [->, line width=0.4mm, style=dashed, bend left=75, looseness=1.10] (1.center) to (2.center);
	\draw [->,line width=0.4mm, bend right=45, looseness=0.70] (3.center) to (4.center);
	\end{tikzpicture}
    \caption{Graphical depiction of a two-step dynamics \(\Dyn = (\cpt_1,\cpt_2)\). This dynamics can be realized with classical memory as in Eq.~\eqref{eq:def_classicalmemory_instruments} if the first map can be written in terms of quantum instruments $\cpt_1 = \sum_\alpha \Inst_\alpha$ whose measurement outcome \(\alpha\) conditions the choice of the map $\mathcal{F}_\alpha$ in the second step in order to obtain the overall map $\cpt_2=\sum_\alpha \mathcal{F}_\alpha \circ \inst_\alpha$.}
    \label{fig:evolution}
\end{figure}
If such a decomposition does not exist, the dynamics is said to require \emph{truly quantum memory} (see Ref.~\cite{BaeBeyStr2024} for details).
Based on this definition the following theorem can be formulated:
\begin{theorem}
	\label{th:identifynMU}
	A unital map $\cptu$ is mixed-unitary if and only if the dynamics $\Dyn=\dynamics{\cptu, \id}$ is realizable with classical memory.
\end{theorem}
\begin{proof}
For the special case of a dynamics \(\Dyn=\dynamics{\cpt_1, \cpt_2}\) with \(\cpt_1 = \cptu\) and \(\cpt_2 = \id\), the definition of classical memory in Eq.~\eqref{eq:def_classicalmemory_instruments} reduces to a form that is identical to the existence of an environment-assisted error correction scheme in the sense of Ref.~\cite{GreWer2003} whose equivalence to \(\cptu\) being MU has been established in said reference (see App.~\ref{app:proof_of_th1} for details and App.~\ref{app:corrolary} for a corollary).
\end{proof}

Intuitively, the connection can be understood by the fact that if $\cpt_1$ is a mixture of unitary channels then the $\mathcal{F}_\alpha$ can be chosen to be the corresponding inverse unitary operations, thus restoring the initial quantum state. 
A correction scheme as in Ref.~\cite{GreWer2003} exists if and only if the measurement outcomes can be used to perfectly reverse the channel \(\cpt\), i.e.\ if the second map in the dynamics is the identity.

\section{SDP non-MU witnesses}
Having established the equivalence between \(\cptu\) being non-MU and \(\Dyn = (\cptu,\id)\) being a dynamics that requires quantum memory, we can employ the existing criteria for the latter to demonstrate the former. 

The Choi state \(E\) of a completely positive map \(\cpt\) is defined through the Choi-Jamiołkowski isomorphism 
\begin{align}
\label{eq:choi-isomorphism}
    E = (\id \otimes \cpt)\ketbra{\Phi^+}{\Phi^+}, 
\end{align}
with the unnormalized Bell state \(\ket{\Phi^+} = \sum_{k=0}^{d-1} \ket{k}\otimes\ket{k}\).
A dynamics \(\Dyn=(\cpt_1,\cpt_2)\) can always be regarded as emerging from a process tensor in Choi representation \( X^{AB'BC'}\)~\cite{pollockNonMarkovianQuantumProcesses2018}.
Process tensors are the generalization of quantum channels to multiple times. They describe the evolution of an open quantum system under the influence of local interventions on this system at distinct times. We consider an initial state preparation, one intervention, and a final output.
Thus, here the Hilbert space \(A\) describes the input at time \(t_0\), \(B'\) is the output Hilbert space at time \(t_1\), \(B\) is the input Hilbert space for the second time step, and \(C'\) the output Hilbert space at time \(t_2\).
For \(X\) to be compatible with the dynamics, it must reproduce the correct Choi states \(E_{1,2}\) of the maps \(\cpt_{1,2}\)~\cite{chiribellaTransformingQuantumOperations2008,chiribellaQuantumCircuitArchitecture2008}:
\begin{gather}
    \tr_{C'}[X^{AB'BC'}] = E_1 \otimes \id_B, \\
    \bra{\Phi^+_{B'B}} X^{AB'BC'} \ket{\Phi^+_{B'B}} = E_2.
\end{gather}
In general, there are infinitely many \(X\) that yield the dynamics \(\Dyn=(\cpt_1,\cpt_2)\). 
However, if a dynamics admits a classical memory decomposition as in Eq.~\eqref{eq:def_classicalmemory_instruments}, 
there is at least one separable realization of the form~\cite{YuOhsNguNim2025:p}
\begin{align}
    X^{AB'BC'} = \sum_\alpha I_\alpha^{AB'} \otimes F_\alpha^{BC'}, 
\end{align}
where \(I_\alpha\) and \(F_\alpha\) are the Choi states of \(\Inst_\alpha\) and \(\mathcal{F}_\alpha\), respectively [see Eq.~\eqref{eq:choi-isomorphism}].

Accordingly, if a unital map \(\cptu\) is MU, there is a process tensor compatible with the dynamics \(\Dyn = (\cptu, \id)\) that is separable with respect to the partition \(AB'|BC'\).
Deciding separability is a hard problem. However, the non-existence of a separable representation can often be demonstrated by numerically efficient semidefinite programs (SDP) subject to a positive partial transpose (PPT) constraint.
This method was employed in Ref.~\cite{YuOhsNguNim2025:p} for the case of quantum memory detection and we will make use of this criterion here as well.
Adapted for our case of non-MU detection we can formulate the following SDP:
\begin{align}
    \wit = &\max_X \left(\frac{1}{d^2}\bra{\Phi^+} G \ket{\Phi^+} - 1 \right),  \label{eq:objective} \\
  \text{s.t.:\quad }      & X^{AB'BC'} \succeq 0,\;  (X^{AB'BC'})^{\top_{BC'}} \succeq 0, \label{eq:PPT-constraint}\\
        & \tr_{C'}[X^{AB'BC'}] = E_u \otimes \id_B, \label{eq:first-marginal} \\
    &G = \bra{\Phi^+_{B'B}} X^{AB'BC'} \ket{\Phi^+_{B'B}} \label{eq:second:marginal}.
\end{align}
A unital map $\cptu$ is non-MU, if for this witness $\wit$ we find
\begin{align}
    \wit < 0.
\end{align}
The constraints describe the positivity of \(X\) and its partial transpose~\eqref{eq:PPT-constraint}, the marginal for the first channel which has to agree with the unital channel \(\cptu\) whose non-MU property we are interested in~\eqref{eq:first-marginal}, and the second marginal which we denote by \(G\)~\eqref{eq:second:marginal}.
The fidelity of this marginal \(G\) with the expected Choi state of the identity map \(E_2 = \ketbra{\Phi^+}{\Phi^+}\) is maximized in the objective function of the SDP~\eqref{eq:objective} (note the scaling by \(d^2\) due to the unnormalized Choi states). 
If the maximum \(\wit\) is negative, the dynamics \(\Dyn = (\cptu, \id)\) cannot be realized with a separable process tensor \(X\), hence, \(\cptu\) is of non-MU form. 
Thus, \(\wit\) defines a non-MU witness.

The PPT  constraint in Eq.~\eqref{eq:PPT-constraint} leads to a sufficient but not necessary witness. 
Stronger witnesses can be obtained by adding further PPT constraints for symmetric extensions of the Hilbert space of \(X\) to the SDP. In theory such a hierarchy is guaranteed to converge in the limit of infinite depth~\cite{dohertyDistinguishingSeparableEntangled2002,dps_pra}. 
However, already the first order, given above, often provides a strong witness, as we show in the following sections.

\section{The Landau-Streater channel ($d=3$)}
It is known that unital qutrit dynamics of Kraus rank $r=2$ are always MU~\cite{LanStr1993}.
For channels of rank $r\geq 3$ this is no longer true.
The prototypical non-MU qutrit channel
is the \emph{Landau-Streater} (LS) channel $\cptls$, given by the Kraus representation
\(K_1 = \ketbra{1}{2}\!-\!\ketbra{2}{1},
    K_2 = \ketbra{0}{2}\!-\!\ketbra{2}{0},
    K_3 = \ketbra{0}{1}\!-\!\ketbra{1}{0}\)~\cite{kummererEssentiallyCommutativeDilations1987, LanStr1993}.
Mixing this channel with the identity channel, we define a family of channels
\begin{align}
    \label{eq:family_ls_id}
    \cpt_p = p \id + (1-p) \cptls, \quad p\in[0,1].
\end{align}
In Ref.~\cite{kribsOperatorAlgebraGeneralization2024a} it was shown analytically that this family admits a MU representation only in the trivial case $p=1$.
This result serves as a first benchmark for our quantum memory-based non-MU witness. 
In Fig.~\ref{fig:family_ed_id.pdf} we plot the optimized fidelity \(\wit\) [see Eq.~\eqref{eq:objective}] as a function of the mixing parameter \(p\) in Eq.~\eqref{eq:family_ls_id}, showing that the witness correctly detects the non-MU behavior for all \(p<1\). 
\begin{figure}
    \centering
    \includegraphics[width=0.72\linewidth]{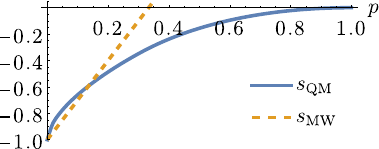}
    \caption{Witnesses of non-mixed-unitaryness in the family of CPT maps given by Eq.~\eqref{eq:family_ls_id}. The blue solid line represents the quantum-memory-based witness $\wit$ and indicates non-mixed-unitaryness if it is negative. The orange dashed line computed the witness $s_\mathrm{MW}$ according Eq.~\eqref{eq:mendlwolf} indicates non-mixed-unitaryness if the value is below zero, which is the case for all $p<1/3$. Both witnesses are normalized to $s=-1$ at $p=0$.}
    \label{fig:family_ed_id.pdf}
\end{figure}

Criteria for the non-MU nature of unital channels often only apply to a certain specific class with high symmetry. 
A generally applicable witness for non-MU qutrit channels was proposed by Mendl and Wolf~\cite{MenWol2009}. They show that a unital qutrit channel \(\cptu\) is non-MU if its Choi state \(E_u\) satisfies \(s_\mathrm{MW}<0\), where
\begin{align}
    \label{eq:mendlwolf}
    s_\mathrm{MW} = \tr[M E_u],&&M = (3\mathbb{F}  + \id)/3,
\end{align}
and \(\mathbb{F}: \ket{k,l} \mapsto \ket{l,k}\).
The Mendl-Wolf criterion serves as a numerical benchmark for our witness \(\wit\). We plot \(s_\mathrm{MW}\) in Fig.~\ref{fig:family_ed_id.pdf} for comparison. The Mendl-Wolf witness detects the true nature of the family only for $p < 1/3$. Thus, our witness performs  significantly better for the channels $\cpt_p$ in Eq.~\eqref{eq:family_ls_id}. 

\bigskip

\section{Convex mixtures of extremal channels}
For the 80-dimensional space of qutrit channels it is known that unital channels form a convex subset and therein MU channels form another convex subset of non-zero volume~\cite{watrousMixingDoublyStochastic2009}.
By contrast, the subset of non-MU channels is not convex (see App.~\ref{app:non-convex}).

\begin{figure}
    \centering
    \includegraphics[]{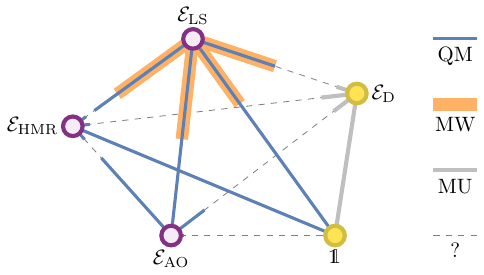}
    \caption{Overview of the performance of the quantum memory witness (QM) and the Mendl-Wolf witness (MW) for convex combinations of different unital channels. The three channels $\cptls$, $\cptao$ and $\cpthmr$ are verifiably non-MU while the fully depolarizing channel $\cptd$ and the identity channel $\id$ have known MU decompositions. It can be seen that for the convex combination of $\cptls$ and $\cptd$  the QM witness (blue thin line) performs equally well as the Mendl-Wolf-witness (orange thick line), while it outperforms the latter one in every other of the investigated convex combinations. The solid gray lines indicate parameter ranges for which it is known that the channel is MU \cite{watrousMixingDoublyStochastic2009} while dashed gray lines indicate parameter ranges for which nothing can be concluded from the witnesses. The details and  parameter ranges depicted in this figure can be found in App.~\ref{app:performance} in Tab.~\ref{tab:overview_performance}.}
    \label{fig:relations_channels}
\end{figure}

To demonstrate the versatility of the quantum memory-based approach to certifying non-MU channels, we analyze pairwise convex combinations of five different unital qutrit channels.
Two of them, the identity map $\id$ and the fully depolarizing channel $\cptd$, are of MU form. The other three are known to be non-MU: the Landau-Streater channel $\cptls$ \cite{LanStr1993, kummererEssentiallyCommutativeDilations1987}, the Arveson-Ohno channel $\cptao$ \cite{ohnoMaximalRankExtremal2010, haagerupExtremePointsFactorizability2021}, and one channel $\cpthmr$ of the Haagerup-Musat-Ruskai family \cite{haagerupExtremePointsFactorizability2021}.
Apart from the depolarizing channel \(\cptd\), all of them are extreme points of the set of unital channels (see Fig.~\ref{fig:geometry-qutrit-channels}). 

The fully depolarizing channel of a qutrit $\cptd\left[\rho\right]={\tr{\rho}} \id/3$ maps every state to the maximally mixed state. It is known that mixing this channel with any other unital channel as in Eq.~\eqref{eq:family_ls_id} is MU for $p\geq {7}/{8}$ in the qutrit case \cite{watrousMixingDoublyStochastic2009}.
The Landau-Streater channel \cite{kummererEssentiallyCommutativeDilations1987, LanStr1993}, which, in the qutrit case, is related to the Werner-Holevo channel \cite{KarimipourLSandWH}, is a unital qutrit channel without MU representation. It is not only extremal in the set of unital channels but also in the set of all qutrit channels.

There are also unital qutrit channels without MU representations which are extremal unital channels but not extremal in the set of all channels. Two known examples are the Arveson-Ohno channel $\cptao$ and a family constructed by Haagerup, Musat, and Ruskai. Both have in common that they are rank-4 channels. One representation  of the Arveson-Ohno channel $\cptao$ is given by~\cite{ohnoMaximalRankExtremal2010}
\begin{align}
    K_1 &= \ketbra{0}{0}, \quad K_2 = \ketbra{0}{1}+\sqrt{2}\ketbra{1}{2} \\
    K_3 &= \sqrt{2}\ketbra{1}{0}+\sqrt{3}\ketbra{2}{1}, \quad K_4 = \ketbra{2}{0} + \sqrt{2} \ketbra{0}{2}. \notag
\end{align}
The family $\cpthmr(\alpha, \beta)$ is non-MU for almost all $\alpha, \beta \in \mathbb{C}$ with $\abs{\alpha}^2+\abs{\beta}^2=1$. A Kraus representation is given by~\cite{haagerupExtremePointsFactorizability2021}
\begin{align}
    K_1 &= \alpha \ketbra{0}{0}+\ketbra{1}{2}, \quad K_2=\beta \ketbra{0}{2} + \ketbra{2}{1} \\
    K_3 &= -\ketbra{0}{1}-\beta^* \ketbra{2}{0}, \quad K_4 = \ketbra{1}{0}+\alpha^* \ketbra{2}{2}.\notag
\end{align}
The parameters chosen for the investigation in Tab.~\ref{tab:overview_performance} are $\alpha={1}/{\sqrt{5}}$ and $\beta={2}/{\sqrt{5}}$. They are chosen such that the absolute value of $\wit$ according to the SDP in Eq.~\eqref{eq:objective} is maximal.

We will now investigate all possible convex combinations of two of such channels using our non-MU witness.
For most of those convex combinations no analytical results exist but we compare the performance of our witness \(\wit\) in Eq.~\eqref{eq:objective} with the witness \(s_\text{MW}\) from Eq.~\eqref{eq:mendlwolf}.
The results are visualized in Fig.~\ref{fig:relations_channels}. 
Both witnesses detect the same parameter range for the non-MU behavior of the channel $\cptls$ mixed with a fully depolarizing channel \(\cptd\). 
For all other convex combinations, the quantum-memory-based witness \(\wit\) is able to demonstrate the non-MU property in significantly  larger parameter regions than the witness~\(s_\text{MW}\).

\section{Detection of genuine quantum decoherence}
For a single qubit, every unital channel is MU and therefore qubit decoherence processes can in principle always be corrected by a measure and feedback scheme on the environment~\cite{GreWer2003}. 
This is no longer true for a multi-qubit system that couples jointly to its surrounding.
Here we analyze the non-MU detection through our quantum-memory-based witness for the prototypical case of this class: a two-qubit system that experiences a joint dephasing dynamics. 

 \begin{figure}
    \hspace{2mm}
     \includegraphics[width=.9\linewidth]{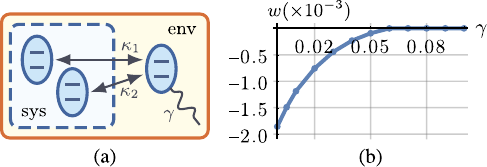}
    \caption{(a) Model for a joint dephasing of a two-qubit system. The system qubits interact with a single environment qubit through the Hamiltonian in Eq.~\eqref{eq:twoqubit-unital}. The environment qubit undergoes a Markovian damping with strength \(\gamma\). (b) Value of the witness \(w\) in Eq.~\eqref{eq:Hermitian-witness-value} at time \(t=1.325\) as a function of the damping strength \(\gamma\). Model parameters as in Fig.~\ref{fig:two-qubit-model-memory}. The operator \(W\) has been constraint to 28 independent terms which makes an experimental implementation significantly easier than a full channel tomography (see App.~\ref{app:relaxed-witness}).}
    \label{fig:two-qubit-model-memory-PQR}
\end{figure}

From an experimental point of view, the witness \(\wit\) described in Eq.~\eqref{eq:objective} requires a full tomography of the unital channel \(\cptu\) as an input.
For the two-qubit case this could for example be achieved by measuring all correlators
\begin{align}
\label{eq:correlator-in-channel-form}
    \!C_{ijkl} \!=\! \tr\!\left[\cptu[\sigma_i\!\otimes\!\sigma_j] \sigma_k\!\otimes\!\sigma_l\right]\!=\! \tr\!\left[\!E_u \sigma_i^\top\!\!\!\otimes\!\sigma_j^\top\!\!\!\otimes\!\sigma_k\!\otimes\!\sigma_l\right]\!,
\end{align}
with \(i,j,k,l \in \{0,1,2,3\}\), and \(\sigma_i\) are the Pauli operators with \(\sigma_0 = \id\).
However, if we are only interested in demonstrating the non-MU nature of a channel---for example in order to rule out the existence of an error mitigation scheme---this complete characterization of the channel is unnecessary.
Instead, it often suffices to measure only some correlators \(C_{ijkl}\)~\cite{YuOhsNguNim2025:p}.
This is particularly important for systems of higher dimension, as a full tomography becomes exponentially expensive with growing dimension. Thus, keeping the number of measured correlators, and therefore also the number of variables in the SDP, low is necessary for a practical implementation. Clearly, such a restricted witness is by construction weaker than the full one and its performance will very much depend on the choice of measured correlators and the given dynamics under investigation.

Let us define a Hermitian operator
\begin{align}
\label{eq:Hermitian-witness-decomposition}
    W = \sum_{i,j,k,l=0}^3 w_{ijkl} \sigma_i^\top\otimes\sigma_j^\top\otimes\sigma_k\otimes\sigma_l, \quad w_{ijkl} \in \mathbb{R}. 
\end{align}
It follows directly from Thm.\,12 in Ref.~\cite{YuOhsNguNim2025:p} and our Thm.\,1  that this \(W\) together with another Hermitian operator \(P\) defines a valid non-MU witness for a two-qubit channel
if there are two positive semidefinite operators \(Q\) and \(R\) such that
\begin{align}
    \frac{1}{d}W^{AB'} + P^{AC'} \otimes \ketbra{\Phi^+}{\Phi^+}^{B'B} =  Q + R^{\top_{BC'}}.
\end{align}
The value of the witness is defined by
\begin{align}
\label{eq:Hermitian-witness-value}
    w = \tr[W E_u] + \bra{\Phi^+}P\ket{\Phi^+},
\end{align}
and the unital channel \(\cptu\) is non-MU if we find \(w < 0\). The second term Eq.~\eqref{eq:Hermitian-witness-value} is a real number fixed by \(P\). The first term can experimentally be determined by measuring the correlators \(C_{ijkl}\) in Eq.~\eqref{eq:correlator-in-channel-form} for a given channel \(\cptu\) and weighting them by the numbers \(w_{ijkl}\) in Eq.~\eqref{eq:Hermitian-witness-decomposition}.

Suitable operators \(W\) and \(P\) can be found by an SDP if the experimentally expected channel \(\cptu\) is known. 
An experimental verification of the non-MU nature would then only require an implementation of the observable \(W\) to determine \(w\) in Eq~\eqref{eq:Hermitian-witness-value}.
We show in the following that this approach is viable for the detection of non-MU dynamics in a two-qubit dephasing process.

A model that produces such an evolution is shown in Fig.~\ref{fig:two-qubit-model-memory-PQR}\,(a). 
The two system qubits couple to a single environmental qubit which in turn is damped by a Markovian channel.
The global evolution of an initial product state $\rho = \rho_\sys \otimes \rho_\env$ 
is determined by the three-qubit Hamiltonian \(H = H_\text{int} + H_\text{env}\) with
\begin{align}
    \label{eq:twoqubit-unital}
    H_\mathrm{int} &= \kappa_1 (\sigma_z^{\sys,1} \otimes \sigma_z^\env) +\kappa_2 (\sigma_z^{\sys,2} \otimes \sigma_z^\env)
\end{align}
and \(H_\mathrm{env} = \vec{\Gamma} \cdot\vec{\sigma}\), where $\vec{\sigma} =(\sigma_x, \sigma_y, \sigma_z)$.
Here, the $\kappa_i$ are the interaction strengths between system qubit \(i\) and the environmental qubit, and the $\Gamma_i\in \mathbb{R}$ determine the local Hamiltonian on the environmental qubit.
The latter is damped by a Lindblad operator $L=\sigma_-=\ketbra{0}{1}$ with strength $\gamma$~\cite{GorKosSud1976, Lin1976}.
The three-qubit dynamics is thus determined by \(\dot{\rho}=-\i \comm{H}{\rho} + \mathcal{L}\left[\rho\right],\)
with \(\mathcal{L}\left[\rho\right] = \gamma\left[L \rho L^\dagger/2 - \left(L^\dagger L \rho + \rho L^\dagger L \right) \right]\). Tracing over the environmental qubit, this leads to a 
time-continuous non-Markovian unital decoherence dynamics \(\cptu(t)\) of the two-qubit system.

\begin{figure}
    \centering
    \includegraphics[width=.86\linewidth]{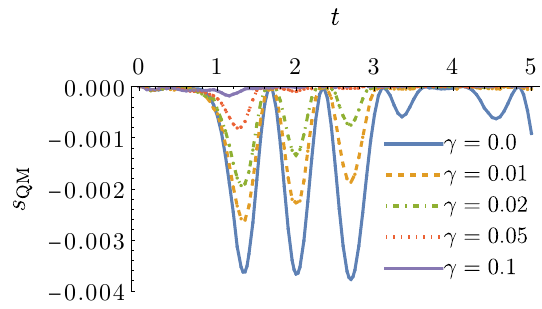}
    \caption{Non-MU witness from Eq.~\eqref{eq:objective} applied to the joint dephasing dynamics generated by the model in Fig.~\ref{fig:two-qubit-model-memory-PQR}\,(a). The parameters $\kappa_2=0.4$, $\kappa_2=1.2$, $\Gamma_x=0.4$, $\Gamma_y=0$, $\Gamma_z=1.0$, and the initial \(\rho_\text{env} = \ketbra{1}{1}\) are chosen such that the dynamics has no MU representation~\cite{HelStr2009}.}
    \label{fig:two-qubit-model-memory}
\end{figure}

Applying the criterion in Eq.~\eqref{eq:objective}, one can verify that \(\cptu(t)\) has no MU representation for a large parameter range. We plot the numerical results for example parameters in Fig.~\ref{fig:two-qubit-model-memory}.
Note that an extension of the Mendl-Wolf criterion in Eq.~\eqref{eq:mendlwolf} to even dimensions leads to a trivial witness \(s_{\text{MW}}\) not plotted here~\cite{MenWol2009}.

The general form of \(W\) in Eq.~\eqref{eq:Hermitian-witness-decomposition} contains \(4^4 = 256\) linearily independent terms.
However, crucially, we can reduce the number of terms significantly while still obtaining a strong witness, as shown in Fig.~\ref{fig:two-qubit-model-memory-PQR}\,(b), where we plot \(w\) in Eq.~\eqref{eq:Hermitian-witness-value} as a function of the damping constant \(\gamma\) for a witness operator \(W\) restricted to 28 independent terms. 
We choose \(t=1.325\) where the MU violation is particularly pronounced for all damping rates (see Fig.~\ref{fig:two-qubit-model-memory}). 
\(W\) and \(P\) are optimized by an SDP for the expected channel \(\cptu\) and we will now show how they can be obtained.

\subsection{How to find \(W\) and \(P\)}
\label{app:relaxed-witness}

Adapting results for quantum memory witnesses~\cite{YuOhsNguNim2025:p}, we obtain suitable \(W\) and \(P\) in Eq.~\eqref{eq:Hermitian-witness-value} for a given unital channel \(\cptu^*\) by the following SDP:
\begin{align}
\label{eq:cheap-witness}
    w^* =& \min_{W,P} \tr[W E^*_u] + \bra{\Phi^+}P\ket{\Phi^+} \\
    &\text{s.t.: }
    W^{AB'}\otimes \frac{\id^{BC'}}{d} + P^{AC'} \otimes \ketbra{\Phi^+}{\Phi^+}^{B'B} \notag \\
    &\qquad\qquad=  Q^{AB'BC'} + (R^{AB'BC'})^{\top_{BC'}}, \notag\\
    &\qquad Q,R \succeq 0, \quad W, P \text{ Hermitian}, \quad \tr[W] = 1,\notag
\end{align}
where the last constraint fixes the normalization of \(W\).
If \(w^*<0\), the channel \(\cptu^*\) with Choi state \(E_u^*\) is non-MU. 
A weaker but computationally even more favorable version \(w^*_R\) can be constructed by setting \(Q=0\) in Eq.~\eqref{eq:cheap-witness}.
By construction, this \(w^*\) and \(w^*_R\) are weaker than \(\wit\) in Eq.~\eqref{eq:objective}. We show its performance on the investigated qutrit channels in Table~\ref{tab:overview_performance}.

Importantly, the operators \(W\) and \(Q\), optimized for a channel \(\cptu^*\) in Eq.~\eqref{eq:cheap-witness}, define a valid witness in the sense of Eq.~\eqref{eq:Hermitian-witness-value} for arbitrary channels \(\cptu\). 
Thus, in an experimental setting, only the Hermitian operator \(W\) needs to be measured in order to confirm the non-MU nature of a unital channel. No full process tomography is necessary to verify the expected channel. Clearly, if the experimentally realized channel does not match the expected channel \(W\) was optimized for, the measurement of \(W\) might not reveal the non-MU nature. 

The form of \(W\) can be constraint for example so as to only contain observables that are measurable in a given experimental setup, or to reduce the number of independent terms in \(W\) [see Eq.~\eqref{eq:Hermitian-witness-decomposition}].
Exploiting the symmetry of the problem, the numerical values in Fig.~\eqref{fig:two-qubit-model-memory-PQR}\,b were obtained by constraining \(W\) to the form \(    W = W_{\textrm{IZ}} + W_{\textrm{XY}}\),
with 
\begin{align}
    W_{\textrm{IZ}} = \sum_{i,j,k,l\in\{0,3\}} a_{ijkl} \sigma_i\otimes\sigma_j\otimes\sigma_k\otimes\sigma_l, 
\end{align}
and \(    W_{\textrm{XY}} = \sum_{i=1}^{12} v_i V_i,\)
with
\begin{align}
    &V_1 = \textrm{IXIX-IYIY+ZXZX+ZYZY}\\
    &V_2 = -\i(\textrm{IIIZ})V_1,\; V_3 = (\textrm{IIZI})V_1,\; V_4 = \i(\textrm{IIZZ})V_1 \notag \\
    &V_5 = \textrm{XIXI+XZXZ-YIYI-YZYZ} \notag\\
    &V_6 = (\textrm{IIIZ})V_5,\; V_7 = \i(\textrm{IIZI})V_5,\; V_8 = \i(\textrm{IIZZ})V_5 \notag \\
    &V_9 = \textrm{XXXX-XYXY-YXYX-YYYY} \notag\\
    &V_{10} = \i(\textrm{IIIZ})V_9,\; V_{11} = \i(\textrm{IIZI})V_9,\; V_{12} = -(\textrm{IIZZ})V_9 \notag,
\end{align}
where \(\textrm{I,X,Y,Z} \equiv \id, \sigma_x, \sigma_y, \sigma_z\). Thus, the SDP in Eq.~\eqref{eq:cheap-witness} optimizes 16 parameters \(a_{ijkl}\) and 12 parameters \(v_i\) instead of the general number of 256 independent \(w_{ijkl}\) in Eq.~\eqref{eq:Hermitian-witness-decomposition}.

{Previous results on demonstrating the truly quantum nature of decoherence processes in $d=4$ relied either on the knowledge of the evolution of the environment or required involved numerical methods~\cite{HelStr2009}. Our SDP approach is numerically straightforward and provides a sufficient witness for non-MU channels that can simplify experimental demonstrations significantly.}

\section{Conclusion}
Determining whether a unital quantum channel admits an MU representation is a problem of both fundamental and practical importance, yet remains hard in general. In this article we have presented an efficient method to rule out MU decompositions by establishing a fundamental connection to quantum memory criteria in non-Markovian dynamics. By mapping the problem to that of identifying separable process tensors, we obtain a hierarchy of semidefinite programs that, in principle, fully characterizes the set of MU channels. Already the lowest level of this hierarchy yields strong witnesses, as demonstrated for multiple prototypical classes of unital qutrit channels, where our approach compares favorably with existing methods.

The connection between quantum memory and mixed-unitarity extends well beyond the specific witnesses constructed here. Any advance in detecting quantum memory---an active area of research---translates directly into new tools for characterizing unital channels, making the framework naturally extensible. 

The witnesses are also readily adaptable to experimental constraints.
For the time-continuous joint dephasing dynamics of a two-qubit system, we devise a simplified form of the witness which can facilitate the experimental verification of non-MU channels without the need for full channel tomography.
As genuinely quantum decoherence processes described by non-MU channels represent a central challenge for noise mitigation in emerging quantum technologies, we expect our framework to find application in the characterization and benchmarking of realistic quantum devices.

\section*{Acknowledgements}
C.\,B.~acknowledges support by the German Academic Scholarship Foundation. K.\,B. acknowledges support by the
NSF under award No.~2239498, NASA under award No.~80NSSC25K7051 and the Sloan Foundation under award No.~G-2023-21102. The Article Processing Charges were funded by the joint publication funds of the TU Dresden and the SLUB Dresden as well as the Open Access Publication Funding of the DFG.

\clearpage

\appendix

\section{Proof of Th.~\ref{th:identifynMU}}
\label{app:proof_of_th1}
\begin{proof}
For the details on the proof of Th.~\ref{th:identifynMU}, let us first compare the definition of a dynamics $\Dyn=\dynamics{\cpt, \cpt_2}$ with CPT maps $\cpt, \cpt_2$ having classical memory with the definition of a correction scheme restoring quantum information for a map $T$. Such a scheme leading to the corrected channel $T_{\mathrm{corr}}$ is defined as \cite{GreWer2003}
    	\begin{align}
    		\label{eq:def_correctionscheme}
    		T_{\mathrm{corr}} = \sum_\alpha R_\alpha \circ T_\alpha, \qquad \sum_\alpha T_\alpha = T,
    	\end{align}
    	where the $R_\alpha$ represent a family of CPT maps.
    	The similarity between Eq.~\eqref{eq:def_classicalmemory_instruments} and Eq.~\eqref{eq:def_correctionscheme} becomes obvious once one identifies
    	\begin{align}
    		\label{eq:cm_gw_identification}
    		\cpt\equiv T, \quad  \cpt_2 \equiv T_{\mathrm{corr}},\quad \mathcal{I}_\alpha \equiv T_\alpha, \quad \mathcal{F}_\alpha \equiv R_\alpha.
    	\end{align}
    	There is hence a one-to-one correspondence between dynamics with classical memory and those with correction schemes.
We will now show that $\cpt$ being MU is the same as  $\Dyn=\dynamics{\cpt, \id}$ having classical memory which again is equivalent to the existence of a scheme restoring quantum information for $\cpt$.

Given some dynamics $\Dyn=\dynamics{\cpt, \id}$ has classical memory, according to the identifications in Eq.~\eqref{eq:cm_gw_identification} this implies that $\cpt_2=T_\mathrm{corr}=\id$. In Ref.~\cite{GreWer2003} a correction scheme restoring quantum information is defined exactly by the fact that $T_\mathrm{corr}=\id$.

Suppose now there is a scheme correcting quantum information for $\cpt$, so $T_\mathrm{corr}=\id$. Then, due to the \emph{no information gain without disturbance}-theorem, every single term in the first sum in Eq.~\eqref{eq:def_correctionscheme} has to be proportional to the identity \cite{GreWer2003}. This can only be satisfied if $K_\alpha^\dagger K_\alpha=p_\alpha \id$ such that the map has to be a convex mixture of unitary dynamics and is hence of the MU type.

Now, finally assume $\cpt$ is MU such that the Kraus operators are proportional to $U_\alpha$.
Choosing $\mathcal{F}_\alpha\left[\rho\right] = U_\alpha^\dagger \rho U_\alpha$ we find
		\begin{align}
			\sum_\alpha \mathcal{F}_\alpha \circ I_\alpha = \sum_\alpha p_\alpha U_\alpha^\dagger \left(U_\alpha \rho U_\alpha^\dagger\right) U_\alpha = \rho,
		\end{align}
where we used that the $U_\alpha$ are unitary and the $p_\alpha$ sum up to one.
It hence follows that for a MU map $\cpt$ the dynamics $\Dyn=\dynamics{\cpt, \id}$ has classical memory \cite{BaeBeyStr2024, GreWer2003}.
\end{proof}

\section{Corollary}
\label{app:corrolary}

Note that Th.~\ref{th:identifynMU} is equivalent to the following Corollary:
\begin{corollary}
    \label{th:corollary}
    A unital map $\cptu$ is mixed-unitary if and only if the dynamics $\Dyn=\dynamics{\cptu, \cpt_0}$ has classical memory for all choices of $\cpt_0$.
\end{corollary}
\begin{proof}
	$\Rightarrow$: This is trivial since  $\id$ is one particular choice for $\cpt_0$.\\
	$\Leftarrow$: This is also quite intuitive. Suppose $\Dyn=(\cpt, \id)$ has classical memory. Then the dynamics $\Dyn = \dynamics{\cpt, \cpt_0}$ also only requires classical memory.
	\begin{align}
		\cpt_0 &= \cpt_0 \circ \id\\
			&= \cpt_0\left[\sum_{i=1}^{r} \mathcal{F}_i \left[K_i \rho K_i^\dagger\right]\right]= \sum_{i=1}^{r} \tilde{\mathcal{F}}_i \left[K_i \rho K_i^\dagger\right].
	\end{align}
	where we used the linearity of CPT maps and defined $\tilde{\mathcal{F}}_i = \cpt_0 \circ \mathcal{F}_i$. This result is precisely the definition of classical memory and thus we have shown that if the identity can be reached with classical memory no further memory is necessary to reach any other map $\cpt_0$.
\end{proof}

\section{Non-convexity of the set of non-MU channels}
\label{app:non-convex}
\begin{proof}
    Let us consider the fully depolarizing channel given in an MU representation $\cptd\left[\rho\right]:=\sum_{i=1}^{9}  U_i \rho U_i^\dagger /9$ with nine unitary operators
    \begin{align}
        U_1 = \id, U_2 =  Y, U_3 =  Z, U_4 =  Y^2, U_5 =  YZ \notag\\
        U_6 =  Y^2 Z, U_7 =  YZ^2, U_8 =  Y^2 Z^2, U_9=Z^2,
    \end{align}
    where
    \begin{align}
        Y=
        \begin{pmatrix}
            0 & 1 & 0\\
            0 & 0 & 1\\
            1 & 0 & 0
        \end{pmatrix},
        \qquad
        Z=
        \begin{pmatrix}
            1 & 0 & 0\\
            0 & \omega & 0\\
            0 & 0 & \omega^2
        \end{pmatrix},
    \end{align}
    with $\omega=\exp({2 \pi \i/3})$.
    We can, however, also describe this channel as a convex combination of non-MU channels.
    Note that whenever $\cpt$ is non-MU, the channel $\cpt'=\cpt \circ \mathcal{U}$ is also non-MU, if $\mathcal{U}$ is a unitary map.
    Using the Landau-Streater channel from Eq.~\eqref{eq:family_ls_id} with \(p=0\) as the channel $\cpt$ we can define \(        \cpt_i = \cpt \circ \mathcal{U}_i \)
    with $\mathcal{U}_i\left[\rho\right]:=U_i \rho U_i^\dagger$.
    The convex mixture of these non-MU $\cpt_i$ with $p_i=1/9$ yields the fully depolarizing channel which is MU:
    \begin{align*}
        \sum_{i}^{9} p_i \cpt_i\left[\rho\right] = \cpt\left[\sum_i p_i U_i \rho U_i^\dagger\right] = \cpt\left[\cptd\left[\rho\right]\right] = \cptd\left[\rho\right].
    \end{align*}
\end{proof}

\section{Performance of different witnesses}
\label{app:performance}

We investigate convex combinations of the form $\cpt_p=p \cpt_A + (1-p) \cpt_B$ with respect to the property of being non-MU. The graphical representation in Fig.~\ref{fig:relations_channels} is based on the numerical results summarized in Tab.~\ref{tab:overview_performance}.
There, in addition to the quantum memory witness \(\wit\) from Eq.~\eqref{eq:objective} and the Mendl-Wolf witness \(s_\text{MW}\) from Eq.~\eqref{eq:mendlwolf} we also show the variants \(w^*\) and \(w^*_R\) of the quantum memory witness in Eq.~\eqref{eq:cheap-witness}. Those require on the one hand less computing time but are on the other hand weaker than the full quantum memory witness. In higher dimensions this may be advantageous.
\begin{table}[ht!]
    \centering
    \begin{tabular}{c|c||c|c|c|c}
         $\cpt_A$ &$\cpt_B$ & $\wit\!<\!0$ & $w^*\!<\!0$ & $w^*_R\!<\!0$ & $s_\mathrm{MW}\!<\!0$ \\
         \hhline{=|=|=|=|=|=}
         $\id$ & $\cptls$ & $p\!<\!1$ & $p\!<\!1$ & $p\!<\!1$ & $p\!<\!0.33$\\
         \hline
         $\cptd$ & $\cptls$ & $p\!<\!0.5$ & $p\!<\!0.5$& $p\!<\!0.02$ & $p\!<\!0.5$ \\
         \hline
         $\cptao$ & $ \cptls$ & $\forall p$ & $p\!<\!0.50$ & $p\!<\!0.08$ & $p\!<\!0.51$ \\
         \hline
         $\cpthmr$& $\cptls$ & $ p\!\notin\!\left[0.89, 0.91\right]$ & $ p\!\notin\!\left[0.89, 0.91\right]$ & $p\!<\!0.07$ &$p\!<\!0.63$ \\
         \hline
         $\id$& $\cptao$ & $p\!<\!0.08$ & $p\!<\!0.08$ &  $p=0$ & - \\
         \hline
         $\cptd$ & $\cptao$ & $p\!<\!0.18$ & $p\!<\!0.18$ &  $p=0$ & - \\
         \hline
         $\cpthmr$ & $\cptao$ & $ p\!\notin\!\left[0.71, 0.94\right]$ & $ p\!\notin\!\left[0.69, 0.96\right]$ & $p=0$ & - \\
         \hline
         $\cptd$ &$\cpthmr$ & $p\!<\!0.03$ & $p\!<\!0.02$ & - & - \\
         \hline
         $\id$ & $\cpthmr$ & $p\!<\!1$ & $p\!<\!1$ & - & - \\
    \end{tabular}
    \caption{Overview of the performance of the different non-MU witnesses for the convex combinations of the five selected unital qutrit channels. The investigated families take the form $\cpt_p = p \cpt_A + (1-p) \cpt_B$ and the cells of the table contain the parameter ranges of $p$ for which $\cpt_p$ being non-MU can be certified with the according witness.}
    \label{tab:overview_performance}
\end{table}

\bibliography{literature}

@article{GreWer2003,
  author =       {Gregoratti, M. and Werner, R. F.},
  title =        {Quantum Lost and Found},
  journal =      {J. Mod. Opt.},
  year =         2003,
  volume =       50,
  pages =        {915--933},
  month =        apr,
  number =       {6-7},
  issn =         {0950-0340, 1362-3044},
  doi =          {10.1080/09500340308234541},
  url =          {http://www.tandfonline.com/doi/abs/10.1080/09500340308234541},
  langid =       {english},
}

@article{leeDetectingMixedunitaryQuantum2020,
  title = {Detecting Mixed-Unitary Quantum Channels Is {{NP-hard}}},
  author = {Lee, Colin Do-Yan and Watrous, John},
  year = 2020,
  month = apr,
  journal = {Quantum},
  volume = {4},
  eprint = {1902.03164},
  primaryclass = {},
  pages = {253},
  issn = {2521-327X},
  doi = {10.22331/q-2020-04-16-253},
  url = {http://arxiv.org/abs/1902.03164},
  urldate = {2025-04-02},
  archiveprefix = {},
  langid = {english}
}

@article{BaeBeyStr2024,
  author =       {B{\"a}cker, Charlotte and Beyer, Konstantin and Strunz,
                  Walter T.},
  title =        {Local Disclosure of Quantum Memory in Non-{Markov}ian
                  Dynamics},
  journal =      {Phys. Rev. Lett.},
  year =         2024,
  volume =       132,
  pages =        060402,
  month =        feb,
  number =       6,
  publisher =    {{American Physical Society}},
  doi =          {10.1103/PhysRevLett.132.060402},
  url =          {https://link.aps.org/doi/10.1103/PhysRevLett.132.060402},
  urldate =      {2024-02-09},
}

@article{LanStr1993,
  author =       {Landau, L. J. and Streater, R. F.},
  title =        {On {Birkhoff}'s Theorem for Doubly Stochastic Completely
                  Positive Maps of Matrix Algebras},
  journal =      {Lin. Alg. Appl.},
  year =         1993,
  volume =       193,
  pages =        {107--127},
  month =        nov,
  issn =         {0024-3795},
  doi =          {10.1016/0024-3795(93)90274-R},
  url =          {https://www.sciencedirect.com/science/article/pii/002437959390274R},
  urldate =      {2023-04-29},
  langid =       {english}
}

@article{ohnoMaximalRankExtremal2010,
  title = {Maximal Rank of Extremal Marginal Tracial States},
  author = {Ohno, Hiromichi},
  year = 2010,
  month = sep,
  journal = {J. Math. Phys.},
  volume = {51},
  number = {9},
  eprint = {},
  primaryclass = {},
  pages = {092101},
  issn = {0022-2488, 1089-7658},
  doi = {10.1063/1.3481567},
  url = {http://arxiv.org/abs/0911.3459},
  archiveprefix = {},
  langid = {english},
}

@article{haagerupExtremePointsFactorizability2021,
  title = {Extreme {{Points}} and {{Factorizability}} for {{New Classes}} of {{Unital Quantum Channels}}},
  author = {Haagerup, Uffe and Musat, Magdalena and Ruskai, Mary Beth},
  year = 2021,
  month = oct,
  journal = {Annales Henri Poincar\'e},
  volume = {22},
  number = {10},
  eprint = {},
  primaryclass = {},
  pages = {3455--3496},
  issn = {1424-0637, 1424-0661},
  doi = {10.1007/s00023-021-01071-y},
  url = {http://arxiv.org/abs/2006.03414},
  archiveprefix = {},
  langid = {english},
}

@article{HelStr2009,
  author =       {Helm, Julius and Strunz, Walter T.},
  title =        {Quantum Decoherence of Two Qubits},
  journal =      {Phys. Rev. A},
  year =         2009,
  volume =       80,
  pages =        042108,
  num_pages =    5,
  month =        oct,
  number =       4,
  publisher =    {{American Physical Society}},
  doi =          {10.1103/PhysRevA.80.042108},
  url =          {https://link.aps.org/doi/10.1103/PhysRevA.80.042108},
  urldate =      {2024-02-10}
}

@article{kummererEssentiallyCommutativeDilations1987,
  title = {The Essentially Commutative Dilations of Dynamical Semigroups {{onMn}}},
  author = {K{\"u}mmerer, Burkhard and Maassen, Hans},
  year = 1987,
  month = mar,
  journal = {Communications in Mathematical Physics},
  volume = {109},
  number = {1},
  pages = {1--22},
  issn = {1432-0916},
  doi = {10.1007/BF01205670},
  url = {https://doi.org/10.1007/BF01205670},
  urldate = {2025-10-12},
  langid = {english}
}

@article{kribsOperatorAlgebraGeneralization2024a,
  title = {Operator Algebra Generalization of a Theorem of {{Watrous}} and Mixed Unitary Quantum Channels},
  author = {Kribs, David W and Levick, Jeremy and Pereira, Rajesh and Rahaman, Mizanur},
  year = 2024,
  month = mar,
  journal = {J. Phys. A-Math.},
  volume = {57},
  number = {11},
  pages = {115303},
  publisher = {IOP Publishing},
  issn = {1751-8121},
  doi = {10.1088/1751-8121/ad2cb0},
  url = {https://doi.org/10.1088/1751-8121/ad2cb0},
  urldate = {2025-10-24},
  langid = {english},
  }

@article{MenWol2009,
  author =       {Mendl, Christian B. and Wolf, Michael M.},
  title =        {Unital Quantum Channels -- Convex Structure and Revivals of
                  {Birkhoff}'s Theorem},
  journal =      {Commun. Math. Phys.},
  year =         2009,
  volume =       289,
  pages =        {1057--1086},
  month =        aug,
  number =       3,
  issn =         {1432-0916},
  doi =          {10.1007/s00220-009-0824-2},
  url =          {https://doi.org/10.1007/s00220-009-0824-2},
  urldate =      {2023-01-25},
  langid =       {english},
  filename =     {MenWol2009.pdf},
}

@article{watrousMixingDoublyStochastic2009,
  title = {Mixing Doubly Stochastic Quantum Channels with the Completely Depolarizing Channel},
  author = {Watrous, J.},
  year = 2009,
  month = may,
  journal = {Quantum Inf. Comput.},
  volume = {9},
  number = {5\&6},
  pages = {406--413},
  issn = {15337146, 15337146},
  doi = {10.26421/QIC9.5-6-4},
  url = {http://www.rintonpress.com/journals/doi/QIC9.5-6-4.html},
  urldate = {2025-11-26},
  langid = {english},
 }

@article{rodriguez-ramosConvexCharacterisationSet2023,
  title = {On the Convex Characterisation of the Set of Unital Quantum Channels},
  author = {{Rodriguez-Ramos}, Constantino and Wilmott, Colin M},
  year = 2023,
  month = oct,
  journal = {J. Phys. A-Math.},
  volume = {56},
  number = {45},
  pages = {455308},
  publisher = {IOP Publishing},
  issn = {1751-8121},
  doi = {10.1088/1751-8121/acfddb},
  url = {https://doi.org/10.1088/1751-8121/acfddb},
  urldate = {2025-09-30},
  langid = {english},
  }

@article{AudSch2008,
  author =       {Audenaert, Koenraad M. R. and Scheel, Stefan},
  title =        {On Random Unitary Channels},
  journal =      {New J. Phys.},
  year =         2008,
  volume =       10,
  note =         { },
  pages =        023011,
  note2 =        { },
  num_pages =    22,
  month =        feb,
  number =       2,
  publisher =    {{IOP Publishing}},
  issn =         {1367-2630},
  doi =          {10.1088/1367-2630/10/2/023011},
  url =          {https://doi.org/10.1088/1367-2630/10/2/023011},
  urldate =      {2022-04-10},
  langid =       {english},
  filename =     {AudSch2008.pdf},
}

@article{GiaCos2021,
  author =       {Giarmatzi, Christina and Costa, Fabio},
  title =        {Witnessing Quantum Memory in Non-{Markov}ian Processes},
  journal =      {Quantum},
  year =         2021,
  volume =       5, 
  pages =        440, 
  num_pages =    12,
  month =        apr,
  publisher =    {{Verein zur F{\"o}rderung des Open Access Publizierens in den
                  Quantenwissenschaften}},
  doi =          {10.22331/q-2021-04-26-440},
  url =          {https://quantum-journal.org/papers/q-2021-04-26-440/},
  urldate =      {2024-01-15},
  langid =       {british},
  filename =     {GiaCos2021.pdf},
}

@article{BaeLinStr2025,
  title = {Verifying quantum memory in the dynamics of spin boson models},
  author = {Bäcker, Charlotte and Link, Valentin and Strunz, Walter T.},
  journal = {Phys. Rev. Res.},
  pages = {--},
  year = {2025},
  month = {Nov},
  publisher = {American Physical Society},
  doi = {10.1103/5bfc-znkj},
  url = {https://link.aps.org/doi/10.1103/5bfc-znkj}
}

@article{BaeBeyStr2025,
  title = {Entropic witness for quantum memory in open system dynamics},
  author = {Bäcker, Charlotte and Beyer, Konstantin and Strunz, Walter T.},
  journal = {Phys. Rev. Res.},
  volume = {7},
  issue = {3},
  pages = {033256},
  numpages = {10},
  year = {2025},
  month = {Sep},
  publisher = {American Physical Society},
  doi = {10.1103/618n-fp8w},
  url = {https://link.aps.org/doi/10.1103/618n-fp8w}
}

@article{MilEglTarThePleSmiHue2020,
  author =       {Milz, Simon and Egloff, Dario and Taranto, Philip and
                  Theurer, Thomas and Plenio, Martin B. and Smirne, Andrea and
                  Huelga, Susana F.},
  title =        {When Is a Non-{Markov}ian Quantum Process Classical?},
  journal =      {Phys. Rev. X},
  year =         2020,
  volume =       10,
  num_pages =    42,
  pages =        {041049},
  month =        dec,
  number =       4,
  publisher =    {{American Physical Society}},
  doi =          {10.1103/PhysRevX.10.041049},
  url =          {https://link.aps.org/doi/10.1103/PhysRevX.10.041049},
  urldate =      {2022-06-26},
}

@article{BanMarHorHor2023,
  author =       {Banacki, Micha{\l} and Marciniak, Marcin and Horodecki, Karol
                  and Horodecki, Pawe{\l}},
  title =        {Information Backflow May Not Indicate Quantum Memory},
  journal =      {Phys. Rev. A},
  year =         2023,
  volume =       107,
  note =         { },
  pages =        032202,
  note2 =        { },
  num_pages =    7,
  month =        mar,
  number =       3,
  publisher =    {American Physical Society},
  doi =          {10.1103/PhysRevA.107.032202},
  url =          {https://link.aps.org/doi/10.1103/PhysRevA.107.032202},
  urldate =      {2024-11-21},
}

@article{TarQuiMurMil2023:p,
  author =       {Taranto, Philip and Quintino, Marco T{\'u}lio and Murao, Mio
                  and Milz, Simon},
  title =        {Characterising the Hierarchy of Multi-time Quantum Processes
                  with Classical Memory},
  journal =      {Quantum},
  year =         2024,
  volume =       8,
  pages =        1328,
  num_pages =    21,
  month =        may,
  publisher =    {Verein zur F{\"o}rderung des Open Access Publizierens in den
                  Quantenwissenschaften},
  doi =          {10.22331/q-2024-05-02-1328},
  url =          {https://quantum-journal.org/papers/q-2024-05-02-1328/},
}

@article{YuOhsNguNim2025:p,
      title={Quantum memory in spontaneous emission processes}, 
      author={Mei Yu and Ties-A. Ohst and Hai-Chau Nguyen and Stefan Nimmrichter},
      year={2025},
      eprint={2504.08605},
      archivePrefix={arXiv},
      primaryClass={quant-ph},
      url={https://arxiv.org/abs/2504.08605}, 
}

@article{BusGanGosBadPanMohDasBer2025,
  author =       {Buscemi, Francesco and Gangwar, Rajeev and Goswami,
                  Kaumudibikash and Badhani, Himanshu and Pandit, Tanmoy and
                  Mohan, Brij and Das, Siddhartha and Bera, Manabendra Nath},
  title =        {Causal and Noncausal Revivals of Information: A New Regime of
                  Non-{Markov}ianity in Quantum Stochastic Processes},
  journal =      {PRX Quantum},
  year =         2025,
  volume =       6,
  note =         { },
  pages =        020316,
  note2 =        { },
  num_pages =    12,
  shorttitle =   {Causal and {{Noncausal Revivals}} of {{Information}}},
  month =        apr,
  number =       2,
  publisher =    {American Physical Society},
  doi =          {10.1103/PRXQuantum.6.020316},
  url =          {https://link.aps.org/doi/10.1103/PRXQuantum.6.020316},
  urldate =      {2025-04-24},
}

@article{dohertyDistinguishingSeparableEntangled2002,
    title = {Distinguishing {Separable} and {Entangled} {States}},
    volume = {88},
    copyright = {http://link.aps.org/licenses/aps-default-license},
    issn = {0031-9007, 1079-7114},
    url = {https://link.aps.org/doi/10.1103/PhysRevLett.88.187904},
    doi = {10.1103/PhysRevLett.88.187904},
    language = {},
    number = {18},
    urldate = {2025-08-13},
    journal = {Phys. Rev. Lett.},
    author = {Doherty, A. C. and Parrilo, Pablo A. and Spedalieri, Federico M.},
    month = apr,
    year = {2002},
    pages = {187904},
}

@article{dps_pra,
  title = {Complete family of separability criteria},
  author = {Doherty, Andrew C. and Parrilo, Pablo A. and Spedalieri, Federico M.},
  journal = {Phys. Rev. A},
  volume = {69},
  issue = {2},
  pages = {022308},
  numpages = {20},
  year = {2004},
  month = {Feb},
  publisher = {American Physical Society},
  doi = {10.1103/PhysRevA.69.022308},
  url = {https://link.aps.org/doi/10.1103/PhysRevA.69.022308}
}

@article{HelStrPRA2010,
  title = {Decoherence and entanglement dynamics in fluctuating fields},
  author = {Helm, Julius and Strunz, Walter T.},
  journal = {Phys. Rev. A},
  volume = {81},
  issue = {4},
  pages = {042314},
  numpages = {9},
  year = {2010},
  month = {Apr},
  publisher = {American Physical Society},
  doi = {10.1103/PhysRevA.81.042314},
  url = {https://link.aps.org/doi/10.1103/PhysRevA.81.042314}
}

@article{BudiniRandomLindblad,
  title = {Random Lindblad equations from complex environments},
  author = {Budini, Adri\'an A.},
  journal = {Phys. Rev. E},
  volume = {72},
  issue = {5},
  pages = {056106},
  numpages = {11},
  year = {2005},
  month = {Nov},
  publisher = {American Physical Society},
  doi = {10.1103/PhysRevE.72.056106},
  url = {https://link.aps.org/doi/10.1103/PhysRevE.72.056106}
}

@article{ChiGou2019ResourceTheories,
  title = {Quantum resource theories},
  author = {Chitambar, Eric and Gour, Gilad},
  journal = {Rev. Mod. Phys.},
  volume = {91},
  issue = {2},
  pages = {025001},
  numpages = {48},
  year = {2019},
  month = {Apr},
  publisher = {American Physical Society},
  doi = {10.1103/RevModPhys.91.025001},
  url = {https://link.aps.org/doi/10.1103/RevModPhys.91.025001}
}

@article{GOUR20151,
title = {The resource theory of informational nonequilibrium in thermodynamics},
journal = {Phys. Rep.},
volume = {583},
pages = {1-58},
year = {2015},
note = {},
issn = {0370-1573},
doi = {https://doi.org/10.1016/j.physrep.2015.04.003},
url = {https://www.sciencedirect.com/science/article/pii/S037015731500229X},
author = {Gilad Gour and Markus P. Müller and Varun Narasimhachar and Robert W. Spekkens and Nicole {Yunger Halpern}},
}

@article{Streltsov_2018,
doi = {10.1088/1367-2630/aac484},
url = {https://doi.org/10.1088/1367-2630/aac484},
year = {2018},
month = {may},
publisher = {IOP Publishing},
volume = {20},
number = {5},
pages = {053058},
author = {Streltsov, Alexander and Kampermann, Hermann and Wölk, Sabine and Gessner, Manuel and Bruß, Dagmar},
title = {Maximal coherence and the resource theory of purity},
journal = {New. J. Phys.},
}

@inbook{Watrous_2018, 
place={Cambridge}, 
title={Unital Channels and Majorization}, 
booktitle={The Theory of Quantum Information}, 
publisher={Cambridge University Press},
author={Watrous, John}, 
year={2018}, 
pages={201–249}}

@article{ContinuousMajorization2025,
  title = {Continuous majorization in quantum phase space for Wigner-positive states and proposals for Wigner-negative states},
  author = {de Boer, Jan and Di Giulio, Giuseppe and Keski-Vakkuri, Esko and Tonni, Erik},
  journal = {Phys. Rev. A},
  volume = {112},
  issue = {3},
  pages = {032405},
  numpages = {30},
  year = {2025},
  month = {Sep},
  publisher = {American Physical Society},
  doi = {10.1103/w561-h3z5},
  url = {https://link.aps.org/doi/10.1103/w561-h3z5}
}

@article{Rastegin_2013,
doi = {10.1088/1742-5468/2013/06/P06016},
url = {https://doi.org/10.1088/1742-5468/2013/06/P06016},
year = {2013},
month = {jun},
publisher = {IOP Publishing and SISSA},
volume = {2013},
number = {06},
pages = {P06016},
author = {Rastegin, Alexey E},
title = {Non-equilibrium equalities with unital quantum channels},
journal = {J. Stat. Mech.: Theory Exp.},
}

@article{AlbashFluctuation2013,
  title = {Fluctuation theorems for quantum processes},
  author = {Albash, Tameem and Lidar, Daniel A. and Marvian, Milad and Zanardi, Paolo},
  journal = {Phys. Rev. E},
  volume = {88},
  issue = {3},
  pages = {032146},
  numpages = {14},
  year = {2013},
  month = {Sep},
  publisher = {American Physical Society},
  doi = {10.1103/PhysRevE.88.032146},
  url = {https://link.aps.org/doi/10.1103/PhysRevE.88.032146}
}

@article{KarimipourLSandWH,
  title = {Noisy Landau-Streater and Werner-Holevo channels in arbitrary dimensions},
  author = {Karimipour, Vahid},
  journal = {Phys. Rev. A},
  volume = {110},
  issue = {2},
  pages = {022424},
  numpages = {12},
  year = {2024},
  month = {Aug},
  publisher = {American Physical Society},
  doi = {10.1103/PhysRevA.110.022424},
  url = {https://link.aps.org/doi/10.1103/PhysRevA.110.022424}
}

@article{TanayTwoQutritAlgorithms2023,
  title = {Two-Qutrit Quantum Algorithms on a Programmable Superconducting Processor},
  author = {Roy, Tanay and Li, Ziqian and Kapit, Eliot and Schuster, DavidI.},
  journal = {Phys. Rev. Appl.},
  volume = {19},
  issue = {6},
  pages = {064024},
  numpages = {18},
  year = {2023},
  month = {Jun},
  publisher = {American Physical Society},
  doi = {10.1103/PhysRevApplied.19.064024},
  url = {https://link.aps.org/doi/10.1103/PhysRevApplied.19.064024}
}

@article{fedorovImplementationToffoliGate2012,
  title = {Implementation of a {{Toffoli}} Gate with Superconducting Circuits},
  author = {Fedorov, A. and Steffen, L. and Baur, M. and {da Silva}, M. P. and Wallraff, A.},
  year = 2012,
  month = jan,
  journal = {Nature},
  volume = {481},
  number = {7380},
  pages = {170--172},
  publisher = {Nature Publishing Group},
  issn = {1476-4687},
  doi = {10.1038/nature10713},
  url = {https://www.nature.com/articles/nature10713},
  copyright = {2011 Springer Nature Limited},
  langid = {english},
}

@article{ringbauerUniversalQuditQuantum2022,
  title = {A Universal Qudit Quantum Processor with Trapped Ions},
  author = {Ringbauer, Martin and Meth, Michael and Postler, Lukas and Stricker, Roman and Blatt, Rainer and Schindler, Philipp and Monz, Thomas},
  year = 2022,
  month = sep,
  journal = {Nature Physics},
  volume = {18},
  number = {9},
  pages = {1053--1057},
  publisher = {Nature Publishing Group},
  issn = {1745-2481},
  doi = {10.1038/s41567-022-01658-0},
  url = {https://www.nature.com/articles/s41567-022-01658-0},
  copyright = {2022 The Author(s), under exclusive licence to Springer Nature Limited},
  langid = {english},
}

@article{lanyonSimplifyingQuantumLogic2009,
  title = {Simplifying Quantum Logic Using Higher-Dimensional {{Hilbert}} Spaces},
  author = {Lanyon, Benjamin P. and Barbieri, Marco and Almeida, Marcelo P. and Jennewein, Thomas and Ralph, Timothy C. and Resch, Kevin J. and Pryde, Geoff J. and O'Brien, Jeremy L. and Gilchrist, Alexei and White, Andrew G.},
  year = 2009,
  month = feb,
  journal = {Nature Physics},
  volume = {5},
  number = {2},
  pages = {134--140},
  publisher = {Nature Publishing Group},
  issn = {1745-2481},
  doi = {10.1038/nphys1150},
  url = {https://www.nature.com/articles/nphys1150},
  copyright = {2009 Springer Nature Limited},
  langid = {english},
}

@article{PavlidisQFTQudit2021,
  title = {Quantum-Fourier-transform-based quantum arithmetic with qudits},
  author = {Pavlidis, Archimedes and Floratos, Emmanuel},
  journal = {Phys. Rev. A},
  volume = {103},
  issue = {3},
  pages = {032417},
  numpages = {12},
  year = {2021},
  month = {Mar},
  publisher = {American Physical Society},
  doi = {10.1103/PhysRevA.103.032417},
  url = {https://link.aps.org/doi/10.1103/PhysRevA.103.032417}
}

@article{gedikComputationalSpeedupSingle2015,
  title = {Computational Speed-up with a Single Qudit},
  author = {Gedik, Z. and Silva, I. A. and {\c C}akmak, B. and Karpat, G. and Vidoto, E. L. G. and {Soares-Pinto}, D. O. and {deAzevedo}, E. R. and Fanchini, F. F.},
  year = 2015,
  month = oct,
  journal = {Scientific Reports},
  volume = {5},
  number = {1},
  pages = {14671},
  publisher = {Nature Publishing Group},
  issn = {2045-2322},
  doi = {10.1038/srep14671},
  url = {https://www.nature.com/articles/srep14671},
  urldate = {2026-03-15},
  copyright = {2015 The Author(s)},
  langid = {english}
}

@article{TS2011EnvironmentAssistedCorrection,
  title = {Environment-assisted error correction of single-qubit phase damping},
  author = {Trendelkamp-Schroer, Benjamin and Helm, Julius and Strunz, Walter T.},
  journal = {Phys. Rev. A},
  volume = {84},
  issue = {6},
  pages = {062314},
  numpages = {7},
  year = {2011},
  month = {Dec},
  publisher = {American Physical Society},
  doi = {10.1103/PhysRevA.84.062314},
  url = {https://link.aps.org/doi/10.1103/PhysRevA.84.062314}
}

@article{suterColloquiumProtectingQuantum2016,
  title = {Colloquium: {{Protecting}} quantum information against environmental noise},
  shorttitle = {Colloquium},
  author = {Suter, Dieter and {\'A}lvarez, Gonzalo A.},
  year = 2016,
  month = oct,
  journal = {Reviews of Modern Physics},
  volume = {88},
  number = {4},
  pages = {041001},
  publisher = {American Physical Society},
  doi = {10.1103/RevModPhys.88.041001},
  urldate = {2026-03-16}
}

@article{laddQuantumComputers2010,
  title = {Quantum computers},
  author = {Ladd, T. D. and Jelezko, F. and Laflamme, R. and Nakamura, Y. and Monroe, C. and O'Brien, J. L.},
  year = 2010,
  month = mar,
  journal = {Nature},
  volume = {464},
  number = {7285},
  pages = {45--53},
  publisher = {Nature Publishing Group},
  issn = {1476-4687},
  doi = {10.1038/nature08812},
  urldate = {2026-03-16},
  copyright = {2010 Macmillan Publishers Limited. All rights reserved}
}

@article{krantz2019,
  title = {A quantum engineer's guide to superconducting qubits},
  author = {Krantz, P. and Kjaergaard, M. and Yan, F. and Orlando, T. P. and Gustavsson, S. and Oliver, W. D.},
  year = 2019,
  month = jun,
  journal = {Applied Physics Reviews},
  volume = {6},
  number = {2},
  pages = {021318},
  issn = {1931-9401},
  doi = {10.1063/1.5089550},
  urldate = {2026-03-16}
}

@article{paladino1NoiseImplications2014,
  title = {1 / f noise: {{Implications}} for solid-state quantum information},
  shorttitle = {1 / f noise},
  author = {Paladino, E. and Galperin, Y. M. and Falci, G. and Altshuler, B. L.},
  year = 2014,
  month = apr,
  journal = {Reviews of Modern Physics},
  volume = {86},
  number = {2},
  pages = {361--418},
  issn = {0034-6861, 1539-0756},
  doi = {10.1103/RevModPhys.86.361},
  urldate = {2026-03-16},
  copyright = {http://link.aps.org/licenses/aps-default-license}
}

@article{goswamiHamiltonianCharacterizationMultitime2025,
  title = {Hamiltonian Characterization of Multi-Time Processes with Classical Memory},
  author = {Goswami, Kaumudibikash and Roy, Abhinash Kumar and Srivastava, Varun and Perez, Barr and Giarmatzi, Christina and Gilchrist, Alexei and Costa, Fabio},
  year = 2025,
  month = nov,
  journal = {New Journal of Physics},
  volume = {27},
  number = {11},
  pages = {114515},
  publisher = {IOP Publishing},
  issn = {1367-2630},
  doi = {10.1088/1367-2630/ae1c64},
  url = {https://doi.org/10.1088/1367-2630/ae1c64},
  langid = {english}
  }

@book{alickiQuantumDynamicalSemigroups2007,
  title = {Quantum {{Dynamical Semigroups}} and {{Applications}}},
  author = {Alicki, Robert},
  year = 2007,
  series = {Lecture {{Notes}} in {{Physics Ser}}},
  edition = {2nd ed},
  number = {v.717},
  publisher = {Springer, Berlin, Heidelberg},
  address = {},
  collaborator = {Lendi, K.},
  isbn = {978-3-540-70860-5 978-3-540-70861-2}
}

@article{pollockNonMarkovianQuantumProcesses2018,
  title = {Non-{{Markovian}} quantum processes: {{Complete}} framework and efficient characterization},
  shorttitle = {Non-{{Markovian}} quantum processes},
  author = {Pollock, Felix A. and {Rodr{\'i}guez-Rosario}, C{\'e}sar and Frauenheim, Thomas and Paternostro, Mauro and Modi, Kavan},
  year = 2018,
  month = jan,
  journal = {Physical Review A},
  volume = {97},
  number = {1},
  pages = {012127},
  issn = {2469-9926, 2469-9934},
  doi = {10.1103/PhysRevA.97.012127},
}

@article{SeifertQuquarts2023,
  title = {Exploring ququart computation on a transmon using optimal control},
  author = {Seifert, Lennart Maximilian and Li, Ziqian and Roy, Tanay and Schuster, David I. and Chong, Frederic T. and Baker, Jonathan M.},
  journal = {Phys. Rev. A},
  volume = {108},
  issue = {6},
  pages = {062609},
  numpages = {13},
  year = {2023},
  month = {Dec},
  publisher = {American Physical Society},
  doi = {10.1103/PhysRevA.108.062609},
  url = {https://link.aps.org/doi/10.1103/PhysRevA.108.062609}
}

@article{GorKosSud1976,
  author =       {Gorini, Vittorio and Kossakowski, Andrzej and Sudarshan,
                  E. C. G.},
  title =        {Completely Positive Dynamical Semigroups of N-level Systems},
  journal =      {J. Math. Phys.},
  year =         1976,
  volume =       17,
  pages =        {821--825},
  month =        may,
  number =       5,
  publisher =    {{American Institute of Physics}},
  issn =         {0022-2488},
  doi =          {10.1063/1.522979},
  url =          {https://aip.scitation.org/doi/abs/10.1063/1.522979}
  }

@article{Lin1976,
  author =       {Lindblad, G.},
  title =        {On the Generators of Quantum Dynamical Semigroups},
  journal =      {Commun. Math. Phys.},
  year =         1976,
  volume =       48,
  pages =        {119--130},
  month =        jun,
  number =       2,
  issn =         {1432-0916},
  doi =          {10.1007/BF01608499},
  url =          {https://doi.org/10.1007/BF01608499},
  urldate =      {2023-04-10},
  langid =       {english}
}

@article{chiribellaTransformingQuantumOperations2008,
    title = {Transforming quantum operations: {Quantum} supermaps},
    volume = {83},
    issn = {0295-5075},
    shorttitle = {Transforming quantum operations},
    url = {https://doi.org/10.1209/0295-5075/83/30004},
    doi = {10.1209/0295-5075/83/30004},
    language = {},
    number = {3},
    urldate = {2026-04-06},
    journal = {Europhysics Letters},
    author = {Chiribella, G. and D'Ariano, G. M. and Perinotti, P.},
    month = jul,
    year = {2008},
    pages = {30004},
}

@article{chiribellaQuantumCircuitArchitecture2008,
    title = {Quantum {Circuit} {Architecture}},
    volume = {101},
    copyright = {http://link.aps.org/licenses/aps-default-license},
    issn = {0031-9007, 1079-7114},
    url = {https://link.aps.org/doi/10.1103/PhysRevLett.101.060401},
    doi = {10.1103/PhysRevLett.101.060401},
    language = {},
    number = {6},
    urldate = {2026-04-06},
    journal = {Physical Review Letters},
    author = {Chiribella, G. and D’Ariano, G. M. and Perinotti, P.},
    month = aug,
    year = {2008},
    pages = {060401},
}

@book{heinosaariMathematicalLanguageQuantum2011,
    address = {Cambridge ; New York},
    title = {The {Mathematical} {Language} of {Quantum} {Theory}: {From} {Uncertainty} to {Entanglement}},
    isbn = {978-0-521-19583-6},
    shorttitle = {The {Mathematical} {Language} of {Quantum} {Theory}},
    language = {},
    publisher = {Cambridge University Press},
    author = {Heinosaari, Teiko and Ziman, Mário},
    year = {2011},
}

@inproceedings{cookComplexityTheoremprovingProcedures1971,
  title = {The Complexity of Theorem-Proving Procedures},
  booktitle = {Proceedings of the Third Annual {{ACM}} Symposium on {{Theory}} of Computing},
  author = {Cook, Stephen A.},
  year = 1971,
  month = may,
  series = {{{STOC}} '71},
  pages = {151--158},
  publisher = {Association for Computing Machinery},
  address = {New York, NY, USA},
  doi = {10.1145/800157.805047},
  url = {https://dl.acm.org/doi/10.1145/800157.805047},
  isbn = {978-1-4503-7464-4},
}

\end{document}